\documentclass[sigconf]{acmart}

\setcopyright{none}

\usepackage{amssymb}
\usepackage{amsmath}
\usepackage{amsfonts}
\usepackage{amsthm}
\usepackage{enumitem}
\usepackage{booktabs}
\usepackage{graphicx}
\usepackage{url}
\usepackage{caption}
\usepackage{subcaption}
\usepackage{float}
\usepackage{graphics}
\usepackage{color}
\usepackage{multirow}
\usepackage{balance}
\usepackage{cleveref}
\usepackage{booktabs}
\usepackage{tikz}
\usepackage[ruled,linesnumbered]{algorithm2e}
\usepackage{cryptocode}
\usepackage{subcaption}
\usepackage{listings}
\usepackage{algorithmic}

\theoremstyle{definition}
\newtheorem{definition}{Definition}
\newtheorem{theorem}{Theorem}

\newcommand{\Post}{P}
\newcommand{\Noise}{T}
\newcommand{\magnitude}{\alpha}
\newcommand{\entropy}{\eta}
\newcommand{\system}{\texttt{MLCapsule}\xspace}
\newcommand{\HW}{{\sf HW}}
\newcommand{\params}{{\sf params}}
\renewcommand{\state}{{\sf state}}
\newcommand{\HWSetup}{\HW.{\sf Setup}}
\newcommand{\HWLoad}{\HW.{\sf Load}}
\newcommand{\HWRun}{\HW.{\sf Run}}

\newcommand{\HWQuoteVerify}{\HW.{\sf QuoteVerify}}

\newcommand{\HWsk}{HW.{\sf sk_{quote}}}
\newcommand{\HWRunQuote}{\HW.{\sf RunQuote}_{{\sf sk_{quote}}}}
\newcommand{\HWT}{{\sf T}}
\newcommand{\hdl}{{\sf hdl}}
\newcommand{\HWin}{{\sf in}}
\newcommand{\HWout}{{\sf out}}
\newcommand{\HWmeta}{{\sf md}_\hdl}
\newcommand{\HWtag}{{\sf tag}_Q}
\newcommand{\sig}{\sigma}
\newcommand{\HWquote}{{\sf quote}}
\newcommand{\aux}{{\sf aux}}
\newcommand{\negl}{{\sf negl}}
\newcommand{\Adv}{{\mathcal A}}
\newcommand{\query}{{\sf query}}
\newcommand{\Enc}{{\sf Enc}}
\newcommand{\Dec}{{\sf Dec}}
\newcommand{\KeyGen}{{\sf KeyGen}}
\newcommand{\sk}{{\sf sk}}
\newcommand{\pk}{{\sf pk}}
\newcommand{\M}{{\mathcal M}}

\newcommand{\Or}{{\mathcal O}}
\newcommand{\Sim}{{\sf Sim}}
\newcommand{\tdata}{{\sf train}_{\sf data}}
\newcommand{\MLmodel}{{\sf ML}_{\sf model}}
\newcommand{\MLdef}{{\sf ML}_{\sf def}}

\newcommand{\MLsk}{{\sf ML}_{\sf sk}}
\newcommand{\MLreq}{{\sf ML}_{\sf req}}
\newcommand{\HMLmodel}{{\sf HML}_{\sf model}}
\newcommand{\Train}{{\sf Train}}
\newcommand{\Classify}{{\sf Classify}}
\newcommand{\Provide}{{\sf Provide}}
\newcommand{\Obtain}{{\sf Obtain}}
\newcommand{\InData}{{\sf input_{data}}}
\newcommand{\OutData}{{\sf output_{data}}}
\newcommand{\command}{{\sf command}}
\newcommand{\data}{{\sf data}}
\newcommand{\out}{{\sf out}}

\newcommand{\ExpModelSec}{{\sf Exp}_{\system,\Adv}^{{\sf secrecy}}}
\newcommand{\ExpIndCPA}{{\sf Exp}_{\Enc,\Adv}}
\newcommand{\exec}{\leftarrow}
\newcommand{\rexec}{\leftarrow_{\$}}

\begin{document}
\title{\system: Guarded Offline Deployment of Machine Learning as a Service
}

\author{Lucjan Hanzlik}
\affiliation{%
  \institution{Stanford University, CISPA Helmholtz Center for Information Security}
}

\author{Yang Zhang}
\affiliation{%
  \institution{CISPA Helmholtz Center for Information Security}
}

\author{Kathrin Grosse}
\affiliation{%
  \institution{CISPA Helmholtz Center for Information Security}
}

\author{Ahmed Salem}
\affiliation{%
  \institution{CISPA Helmholtz Center for Information Security}
}

\author{Max Augustin}
\affiliation{%
  \institution{Max Planck Institute for \\ Informatics}
}

\author{Michael Backes}
\affiliation{%
  \institution{CISPA Helmholtz Center for Information Security}
}

\author{Mario Fritz}
\affiliation{%
  \institution{CISPA Helmholtz Center for Information Security}
}


\begin{abstract}
With the widespread use of machine learning (ML) techniques, ML as a service has become increasingly popular. In this setting, an ML model resides on a server and users can query it with their data via an API. However, if the user's input is sensitive, sending it to the server is undesirable and sometimes even legally not possible. Equally, the service provider does not want to share the model by sending it to the client for protecting its intellectual property and pay-per-query business model.

In this paper, we propose \system, a guarded offline deployment of machine learning as a service. \system executes the model locally on the user's side and therefore the data never leaves the client. Meanwhile, \system offers the service provider the same level of control and security of its model as the commonly used server-side execution. In addition, \system is applicable to offline applications that require local execution. Beyond protecting against direct model access, we couple the secure offline deployment with defenses against advanced attacks on machine learning models such as model stealing, reverse engineering, and membership inference.
\end{abstract}

\maketitle

\section{Introduction}
Machine learning as a service (MLaaS) has become increasingly popular
during the past five years.
Leading Internet companies, such as Google,\footnote{\url{https://cloud.google.com/ml-engine/}}
Amazon,\footnote{\url{https://aws.amazon.com/machine-learning/}}
and Microsoft\footnote{\url{https://azure.microsoft.com/en-us/services/machine-learning-studio/}}
have deployed their own MLaaS.
It offers a convenient way for a service provider to deploy a machine learning (ML) model and equally an instant way for a user/client to make use of the model in various applications.
Such setups range from image analysis over translation
to applications in the business domain.

While MLaaS is convenient for the user, it also comes with several limitations. First, the user has to trust the service provider with the input data. Typically, there are no means of ensuring data privacy and recently proposed encryption mechanisms~\cite{BPTG15} come at substantial computational overhead
especially for state-of-the-art deep learning models containing millions of parameters.
Moreover, MLaaS requires data transfer over the network
which constitutes to high volume communication and provides new attack surface~\cite{MSCS18,OOSF18}.
This motivates us to come up with a client-side solution such that perfect data privacy and offline computation can be achieved.

As a consequence, this (seemingly) comes with a loss of control of the service provider, as the ML model has to be transfered and executed on the client's machine. This raises concerns about revealing details of the model or granting unrestricted access to the user.
The former damages the intellectual property of the service provider,
while the latter breaks the commonly enforced pay-per-query business model.
Moreover, there is a broad range of attack vectors on ML models that raise severe security and privacy risks~\cite{PMSW18}. A series of recent papers have shown different attacks on MLaaS that can lead to reverse engineering
\cite{TZJRR16,OASF18} and training data leakage \cite{FLJLPR14,FJR15,SSSS17,YGFJ18,SZHBFB19}. 
Many of these threats are facilitated by repeated probing of the ML model that the service provider wants to protect against. 
Therefore, we need a mechanism to enforce that the service provider remains in control of the model access as well as provide ways to deploy defense mechanisms to protect the model.

\subsection{Our Contributions}
We propose \system, a guarded offline deployment of machine learning as a service. 
\system follows the popular MLaaS paradigm, but allows for client-side execution where model and computation remain secret.
With~\system, the service provider controls its ML model 
which allows for intellectual property protection and business model maintenance.
Meanwhile, the user gains perfect data privacy and  offline execution, as the data never leaves the client and the protocol is transparent

We assume that the client's platform has access to an Isolated Execution Environment (IEE). \system uses it to provide a secure enclave to run an ML model classification. Moreover, since IEE provides means to prove execution of code, the service provider is assured that the secrets that it sends in encrypted form can only be decrypted by the enclave. This keeps this data secure from other processes running on the client's platform.

To support security arguments about \system, we propose the first formal model for reasoning about the security of local ML model deployment. The leading idea of our model is a property called \emph{ML model secrecy}. This definition ensures that the client can simulate \system using only a server-side API. In consequence, this means that if the client is able to perform an attack against \system, the same attack can be performed on the server-side API.

We also contribute by a proof-of-concept of our solution.
Due to its simplicity and availability we implemented our prototype
on a platform with Intel SGX, despite the fact that the current generation should not
be used due to a devastating attack (see~\cite{BMWGKPSWYS18}).
Note that our solution can be used on any IEE platform for which we can
argue that it implements the abstract requirements defined in section~\ref{sec:secana}.

In more details, in our solution we design so called \system layers, 
which encapsulate standard ML layers and are executed inside the IEE. Those layers 
are able to decrypt (unseal) the secret weight provisioned by the service provider and perform the computation in isolation. 
This modular approach makes it easy to combine layers and form large networks.
For instance, we implement and evaluate the VGG-16~\cite{SZ14} and MobileNet~\cite{HZCKWWAA17} neural networks. In addition, we provide an evaluation of convolution and dense layers and compare the execution time inside the IEE to a standard implementation.

The isolated code execution on the client's platform renders \system
ability to integrate advanced defense mechanism for attacks against machine learning models. 
For demonstration, we propose two defense mechanisms 
against reverse engineering~\cite{OASF18} and membership inference~\cite{SSSS17,SZHBFB19}, respectively,
and utilize a recent proposed defense~\cite{JSDMA18} for model stealing attacks~\cite{TZJRR16}.
We show that these mechanisms can be seamlessly incorporated into~\system,
with a negligible computation overhead,
which further demonstrates the efficacy of our system.

\subsection{Organization}
\autoref{sec:requirement} presents the requirements of
\system. We provide the necessary technical background
in \autoref{sec:back} 
and \autoref{sec:related} summarizes the related work in the field.
In~\autoref{sec:sysdesign}, 
we present \system in detail
and formally prove its security in~\autoref{sec:secana}.
\autoref{sec:implement} discusses the implementation and evaluation of~\system.
We show how to incorporate advanced defense mechanisms in~\autoref{sec:advdef}.
\autoref{sec:discussion} provides a discussion, and the paper is concluded in~\autoref{sec:conclusion}.
\section{Requirements and Threat Model}
\label{sec:requirement}

In this section, 
we introduce security requirements we want to achieve in \system.

\subsection{Model Secrecy and Data Privacy}
\label{sec:basic_requirement}

\noindent\textbf{User Side.}
\system deploys MLaaS locally.
This provides strong privacy guarantees to a user, as her data never leaves her device. Meanwhile, executing machine learning prediction locally avoids the Internet communication between the user and the service provider. 
Therefore, possible attacks due to network communication~\cite{MSCS18,OOSF18} are automatically eliminated. 

\medskip
\noindent\textbf{Server Side.}
Deploying a machine learning model on the client side naively, i.e., providing the trained model to the user as a white box, harms the service provider in the following two perspectives.

\begin{itemize}
\item \emph{Intellectual Property.}
Training an effective machine learning model is challenging, the MLaaS provider needs to get suitable training data 
and spend a large amount of efforts for training the model and tuning various hyperparameters~\cite{WG18}. 
All these certainly belong to the intellectual property of the service provider 
and providing the trained model to the client as a white box will result in the service provider completely losing these valuable information. In this paper we consider the ML model architecture public and only consider the model parameters as private information. 
However, our approach can easily be extended to also protect the model architecture by using tools that protect the
privacy of the code executed inside the IEE (e.g. using~\cite{BWZL18}).
\item \emph{Pay-per-query.}
Almost all MLaaS providers implement the pay-per-query business model. For instance, Google's vision API charges 1.5 USD per 1,000 queries.\footnote{\url{https://cloud.google.com/vision/pricing}}
Deploying a machine learning model at the client side naturally grants a user unlimited number of queries, which breaks the pay-per-query business model.
\end{itemize}

\medskip
To mitigate all these potential damages to the service provider, 
\system needs to provide the following guarantees:
\begin{itemize}
\item Protecting intellectual property
\item Enable the pay-per-query business model
\end{itemize}
In a more general way, we aim for a client-side deployment being indistinguishable from the current server-side deployment.

\subsection{Protection against Advanced Attacks}
Several recent works show that an adversary 
can perform multiple attacks against MLaaS by solely querying its API (black-box access). 
Attacks of such kind include model stealing~\cite{TZJRR16,WG18}, 
reverse engineering~\cite{OASF18}, and membership inference~\cite{SSSS17,SZHBFB19}.
These attacks are however orthogonal to the damages discussed in~\autoref{sec:basic_requirement}, 
as they only need black-box access to the ML model instead of white-box access.
More importantly, it has been shown that
the current MLaaS cannot prevent against these attacks neither~\cite{TZJRR16,SSSS17,OASF18,WG18}.

We consider mitigating these threats as the requirements of \system as well.
Therefore, we propose defense mechanisms against these advanced attacks and show that these mechanisms
can be seamlessly integrated into \system.

\section{Background}
\label{sec:back}

In this section, we focus on the properties of Intel's IEE implementation called Software Guard Extensions (SGX) 
and recall a formal definition of Attested Execution proposed by Fisch et al.~\cite{FVBG17}. We would like to 
stress that \system works with any IEE that implements this abstraction.
We also formalize a public key encryption scheme, which we will use for the concrete instantiation of our system. We stress 

\subsection{SGX}
SGX is a set of commands included in Intel's x86 processor design that allows to create isolated execution environments called enclaves. According to Intel's threat model, enclaves are designed to trustworthily execute programs and handle secrets even if the host system is malicious and the system's memory is untrusted. 

\medskip
\noindent\textbf{Properties. }
There are three main properties of Intel SGX.

\begin{itemize}
\item \emph{Isolation.} Code and data inside the enclave's protected memory cannot be read or modified by any external process. Enclaves are stored in a hardware guarded memory called Enclave Page Cache (EPC), which is currently limited to 128 MB with only 90 MB for the application. Untrusted applications can execute code inside the enclave using entry points called
Enclave Interface Functions \verb+ECALL+s, i.e., untrusted applications can use enclaves as external libraries that are defined by these call functions.
\item \emph{Sealing.} Data stored in the host system is encrypted and authenticated using a hardware-resident key. Every SGX-enabled processor has a special key called Root Seal Key that can be used to derive a so called Seal Key which is specific to the identity of the enclave. This key can then be used to encrypt/decrypt data which can later be stored in untrusted memory. One important feature is that the same enclave can always recover the Seal Key if instantiated on the same platform, however it cannot be derived by other enclaves.
\item \emph{Attestation.} Attestation provides an unforgeable report attesting to code, static data and meta data of an enclave, as well as the output of the performed computation. Attestation can be local and remote. In the first case, one enclave can derive a shared Report Key using the Root Seal Key and create a report consisting of a Message Authentication Code (MAC) over the input data. This report can be verified by a different enclave inside the same platform, since it can also derive the shared Report Key. In case of remote attestation, the actual report for the third party is generated by a so called Quoting Enclave that uses an anonymous group signature scheme (Intel Enhanced Privacy ID~\cite{BL10}) to sign the data.
\end{itemize}

\medskip
\noindent\textbf{Side-channel Attacks. }
Due to its design, Intel SGX is prone to side-channel attacks.
This includes physical attacks (e.g., power analysis), yet successful attacks have not yet been demonstrated.
On the other hand, several software attacks have been demonstrated in
numerous papers~\cite{LSGKKP17,WCPZWBTG17,BMDKCS17}. An attack specifically against secure ML implementations was
presented by Hua et al. \cite{BWZL18}.
Those kinds of attacks usually target flawed implementations and a knowledgeable programmer can write the code in a data-oblivious way, i.e., the software does not have memory access patterns or control flow branches that depend on secret data. In particular, those attacks are not inherent to SGX-like systems~\cite{CLD16}.
Recently, Bulck et al. \cite{BMWGKPSWYS18} presented a devastating attack on SGX that compromises the whole system making
the current generation of the SGX technology useless. Even though the current SGX generation should not be used in practice, 
future instantiations should provide a real-world implementation of the abstract security requirements needed to secure \system.  

\medskip
\noindent\textbf{Rollback.}
The formal model described in the next subsection assumes that the state of the hardware is hidden from the users platform. SGX enclaves 
store encryptions of the enclave's state in the untrusted part of the platform. 
Those encryptions are protected using a hardware-generated secret key, yet this data is provided to the enclave by an untrusted application. Therefore, SGX does not provide any guarantees about freshness of the state and is vulnerable to rollback attacks. Fortunately, there exist hardware solutions relying on counters \cite{SP16}
and distributed software-based strategies \cite{MAKDSGJC17} that can be used to prevent rollback attacks.

\subsection{Definition for SGX-like Hardware}
There are many papers that discuss hardware security models in a formalized way. The general consensus is that those abstractions are useful to formally argue about the security of the system. 

Barbosa et al.~\cite{BPSW16} define a generalized ideal interface to represent SGX-like systems that perform attested computation. A similar model was proposed by Fisch et al. \cite{FVBG17} but was designed specifically to abstract Intel's SGX and support local and remote attestation. Pass, Shi, and Tram{\`{e}}r \cite{PST17} proposed an abstraction of attested execution in the universal composability (UC) model. In this paper we will focus on the formal hardware model by Fisch et al.~\cite{FVBG17}. We decided to use this particular model because it was specifically defined to abstract the features that are supported by SGX which is the hardware used by our implementation. However, since we will only use remote attestation in our instantiation, we omit the local attestation part and refer the reader to the original paper for a full definition. 

Informally, this ideal functionality allows a registered party to install a program inside an enclave, which can then be resumed on any given input. An instance of this enclave possesses internal memory that is hidden from the registering party. However, the main property of attested execution is that the enclave creates an attestation of execution. This attestation provides a proof for third parties that the program was executed on a given input yielding a particular output.

\medskip
\noindent\textbf{Formal Definition. }
We define a secure hardware as follows.
\begin{definition}
A secure hardware functionality $\HW$ for a class of probabilistic polynomial time programs $Q$ consists of the following interface: $\HWSetup$, $\HWLoad$, $\HWRun$, $\HWRunQuote$, $\HWQuoteVerify$. $\HW$ has also an internal state $\state$ that consists of a variable $\HWsk$ and a table $\HWT$ consisting of enclave state tuples indexed by enclave handles. The variable $\HWsk$ will be used to store signing keys and table $\HWT$ will be used to manage the state of the loaded enclave.

\begin{itemize}
\item $\HWSetup(\secparam)$: given input security parameter $\secparam$, it generates the secret key ${\sf sk_{quote}}$ and stores it in $\HWsk$. It also generates and outputs public parameters $\params$.
\item $\HWLoad(\params,Q)$: given input global parameters $\params$ and program $Q$ it first creates an enclave, loads $Q$, and then generates a handle $\hdl$ that will be used to identify the enclave running $Q$. Finally, it sets $\HWT[\hdl] = \emptyset$ and outputs $\hdl$.
\item $\HWRun(\hdl,\HWin)$: it runs $Q$ at state $\HWT[\hdl]$ on input $\HWin$ and records the output $\HWout$. It sets $\HWT[\hdl]$ to be the updated state of $Q$ and outputs $\HWout$.
\item $\HWRunQuote(\hdl,\HWin)$: executes a program in an enclave similar to $\HWRun$ but additionally outputs an attestation that can be publicly verified. 
The algorithm first executes $Q$ on $\HWin$ to get $\HWout$, and updates $\HWT[\hdl]$ accordingly. Finally, it outputs the tuple $\HWquote = (\HWmeta,\HWtag,\HWin,\HWout,\sig)$: $\HWmeta$ is the metadata associated with the enclave, $\HWtag$ is a program tag for $Q$ and $\sig$ is a signature on $(\HWmeta,\HWtag,\HWin,\HWout)$.
\item $\HWQuoteVerify(\params,\HWquote)$: given the input global parameters $\params$ and $\HWquote$ this algorithm outputs $1$ it the signature verification of $\sig$ succeeds. It outputs $0$ otherwise.
\end{itemize}

\medskip
\noindent\textbf{Correctness.} 
A $\HW$ scheme is correct if the following holds.
For all $\aux$, all programs $Q$, all $\HWin$ in the input domain of $Q$ and all handles $\hdl$ we have:
\begin{itemize}
\item if there exist random coins $r$ (sampled in run time and used by $Q$) such that $\HWout = Q(\HWin)$, then
\item $\Pr[\HWQuoteVerify(\params,\HWquote) = 0] = \negl(\lambda)$, where 
$\HWquote = \HWRunQuote(\hdl,\HWin)$.
\end{itemize}

\medskip
\noindent Remote attestation unforgeability is modeled by a game between a challenger $C$ and an adversary $\Adv$.
\begin{enumerate}
\item $\Adv$ provides an $\aux$.

\item $C$ runs $\HWSetup(\secparam,\aux)$ in order to obtain public parameters $\params$, secret key ${\sf sk_{quote}}$ and an initialization string $\state$. It gives $\params$ to $\Adv$, and keeps ${\sf sk_{quote}}$ and $\state$ secret in the secure hardware.

\item $C$ initialized a list $\query = \{\}$.

\item $\Adv$ can run $\HWLoad$ on any input $(\params,Q)$ of its choice and get back $\hdl$.

\item $\Adv$ can also run $\HWRunQuote$ on input $(\hdl,\HWin)$ of its choice and get $\HWquote = (\HWmeta,\HWtag,\HWin,\HWout,\sig)$, where the challenger puts the tuple $(\HWmeta,\HWtag,\HWin,\HWout)$ into $\query$.

\item $\Adv$ finally outputs $\HWquote^* = (\HWmeta^*,\HWtag^*,\HWin^*,\HWout^*,\sig^*)$.
\end{enumerate}

$\Adv$ wins the above game if $\HWQuoteVerify(\params,\HWquote^*) = 1$ and $(\HWmeta^*,  \HWtag^*,\HWin^*\allowbreak ,\HWout^*) \not\in \query$.
The hardware model is remote attestation unforgeable if no adversary can win this game with non-negligible probability.
\end{definition}

\subsection{Public Key Encryption}

\begin{definition}
Given a plaintext space $\M$ we define a public key encryption scheme as a tuple of probabilistic polynomial time algorithms:
\begin{itemize}
\item $\KeyGen(\secparam):$ on input security parameters, this algorithm outputs the secret key $\sk$ and public key $\pk$.
\item $\Enc(\pk,m):$ on input public key $\pk$ and message $m \in \M$, this algorithm outputs a ciphertext $c$.
\item $\Dec(\sk,c):$ on input secret key $\sk$ and ciphertext $c$, this algorithm outputs message $m$.
\end{itemize}

\medskip
\noindent\textbf{Correctness.} A public key encryption scheme is correct if for all
security parameters $\secparam$, all messages $m \in \M$ and all keypairs $(\sk,\pk) = \KeyGen(\secparam)$ we have $\Dec(\sk,\Enc(\pk,m)) = m$.

\medskip
\noindent\textbf{Ciphertext Indistinguishability Against CPA (Chosen Plaintext Attack).} 
\begin{figure}
  \begin{pchstack}[center]
    \procedure{$\ExpIndCPA(\lambda)$}{%
      (\sk,\pk) = \KeyGen(\secparam)\\
      (m_0,m_1) \exec \Adv(\pk); \; b \rexec \{0,1\}\\
      c \exec \Enc(\pk,m_b); \; \hat{b} \exec \Adv(\pk,c) \\
      \pcreturn \hat{b} = b
    }
    \end{pchstack}
  \caption{IND-CPA - security experiment.} \label{fig:indcpa}
\end{figure}
We say that the public key encryption scheme is indistinguishable against chosen plaintext attacks if there exists no adversary for which $\lvert \Pr[\ExpIndCPA = 1] - \frac12 \rvert$ is non-negligible, where experiment $\ExpIndCPA$ is defined in \autoref{fig:indcpa}.
\end{definition}
\section{Related Work}
\label{sec:related}
In this section,
we review related works in the literature. We start by discussing cryptography and ML, concentrating on SGX and homomorphic encryption. We then turn to mechanisms using watermarking to passively protect ML models.
In the end, we describe several relevant attacks against ML models including model stealing, reverse engineering, and membership inference.

\medskip
\noindent\textbf{SGX for ML.} Ohrimenko et al.~\cite{OSFMNVC16} investigated oblivious data access patterns for a range of ML algorithms applied in SGX. Their work focuses on neural networks, support vector machines and decision trees. Additionally,  Hynes et al.~\cite{HCS18}, and Hunt et al.~\cite{HSSSW18} used SGX for ML -- however in the context of training convolutional neural network models. Tram{\`{e}}r et al.~\cite{TB18} introduced Slalom, a system which distributes parts of the computation from SGX to a GPU in order to preserve integrity and privacy of the data. 
Gu et al.~\cite{GHZSLPM18} consider convolutional neural networks and SGX at test time. They propose to split the network, where the first layers are executed in an SGX enclave, and the latter part outside the enclave. The core of their work is to split the network as to prevent the user's input to be reconstructed. 
Compared to \cite{TB18} and \cite{GHZSLPM18}, \system is the first secure and privacy-preserving machine learning system for offline deployment coupled with several defense mechanisms addressing state-of-the-art blackbox-based attacks. 

\medskip
\noindent\textbf{Cryptographic Solutions for ML Models.} Homomorphic encryption has been used to keep both input and result private from an executing server \cite{AB06,BPTG15,DGLLNW16,MZ17,LJLA17}. 
In contrast to our approach, however, homomorphic encryption lacks efficiency, 
especially for deep learning models containing millions of parameters.
Moreover, it does not
allows for implementing transparent mechanisms to defend attacks on the ML model: the user cannot be sure which code is executed. Furthermore, a server-side model protection might require further computations on the encrypted user data and sometimes
also a function of this data in plaintext (e.g. to decide whether to accept a query). 
Lastly, cryptographic based solutions usually do not achieve what we call perfect-privacy (even an unbounded adversary cannot
retrieve any information about the user's data), which is one of the main advantages of \system.

\medskip
\noindent\textbf{Watermarking ML Models.}
Recently, watermarking has been introduced to claim ownership of a trained ML model. Adi et al.~\cite{ABCPK18} propose to watermark a model by training it to yield a particular output for a number of points, thereby enabling identification. Zhang et al.~\cite{ZGJWSHM18}, in contrast, add meaningful content to samples used as watermarks. Watermarking as a passive defense mechanism is opposed to our work in two perspectives.
First, \system deployed on the client-side is indistinguishable from the server-side deployment.
Second, \system allows us to deploy defense mechanisms to actively mitigate advanced attacks against ML models.

\medskip
\noindent\textbf{Model Stealing and Reverse Engineering.} 
Model stealing has been introduced by Tram{\`{e}}r et al.~\cite{TZJRR16}. 
The attacker's goal here is to duplicate a model that is accessible only via an API. Recently, a defense to model stealing called Prada has been proposed by Juuti et al.~\cite{JSDMA18}. We will evaluate Prada's feasibility in \system. Also passive defenses to model stealing have been proposed under the term of watermarking, for example by Adi et al.~\cite{ABCPK18} and Zhang et al.~\cite{ZGJWSHM18}. \system, however, allows to actively mitigate model stealing attacks.
A different line of work aims to infer specific details of the model such as architecture, training method and type of non-linearities, rather than its behavior. Such attacks are called reverse engineering and were proposed by Oh et al.~\cite{OASF18}. We propose a defense against this attack and also show it can be integrated in \system with little overhead.

\medskip
\noindent\textbf{Membership Inference.}
In the setting of membership inference, the attacker aims to decide whether a data point
is in the training set of the target ML model.
Shokri et al.\ are among the first to perform effective membership inference against ML models~\cite{SSSS17}.
Recently, several other attacks have been proposed~\cite{YGFJ18,HMDC17,SZHBFB19}.
In this paper, we propose a defense mechanism
to mitigate the membership privacy risks which can be easily implemented
in~\system.

\section{\system} \label{sec:sysdesign}

In existing MLaaS system, the service provider gives the user access to an API, which she can then use for training and classification. This is a classic client-server scenario, where the service is on the server's side and the user is transferring the data. 
We focus on the scenario where the service provider equips the user with an already trained model and classification is done on the client's machine. 
In this scenario, the trained model is part of the service. This introduces the previously discussed problems, like for example protecting the IP of the service provider or the user's privacy. 
In this section we introduce \system to tackle those issues and then argue how the requirements are met. We would like to stress that the main advantage of \system is its simplicity.

\subsection{Overview}
We start with an overview of the participants and then introduce \system with its different execution phases.

\medskip
\noindent\textbf{Participants.} 
In \system, we distinguish two participants of the system. On the one hand, we have the service provider (SP) that possesses private training data that it uses to train an ML model. On the other hand, we have users that want to use the pre-trained model as previously discussed in~\autoref{sec:requirement}.

We focus for the rest of the paper on deep  networks, as they have recently drawn attention due to their good performance. Additionally, their size make both implementation and design a challenge. Yet, \system generalizes to other linear models, as these can be expressed as a single layer of a neural network.

We consider the design of the applied network to be publicly known. The service provider's objective is to protect the trained weights and biases of all layers. 

\medskip
\noindent\textbf{Approach.}
We now describe \system, which fulfills the requirements from \autoref{sec:requirement}.
To start, we assume that all users have a platform with an isolated execution environment. Note that this is a viable assumption: Intel introduced SGX with the Skylake processor line. Additionally, the latest Kaby Lake generation of processors supports SGX. It is further reasonable to assume that Intel will provide support for SGX with every new processor generation and over time every PC with an Intel processor will have SGX.

The core idea of \system is to leverage the properties of the IEE to ensure that the user has to attest that the code is running in isolation before it is provided the secrets by the SP. This step is called the \emph{setup phase} and
is depicted in \autoref{fig:scheme}.
\begin{figure} 
\includegraphics[width=\columnwidth]{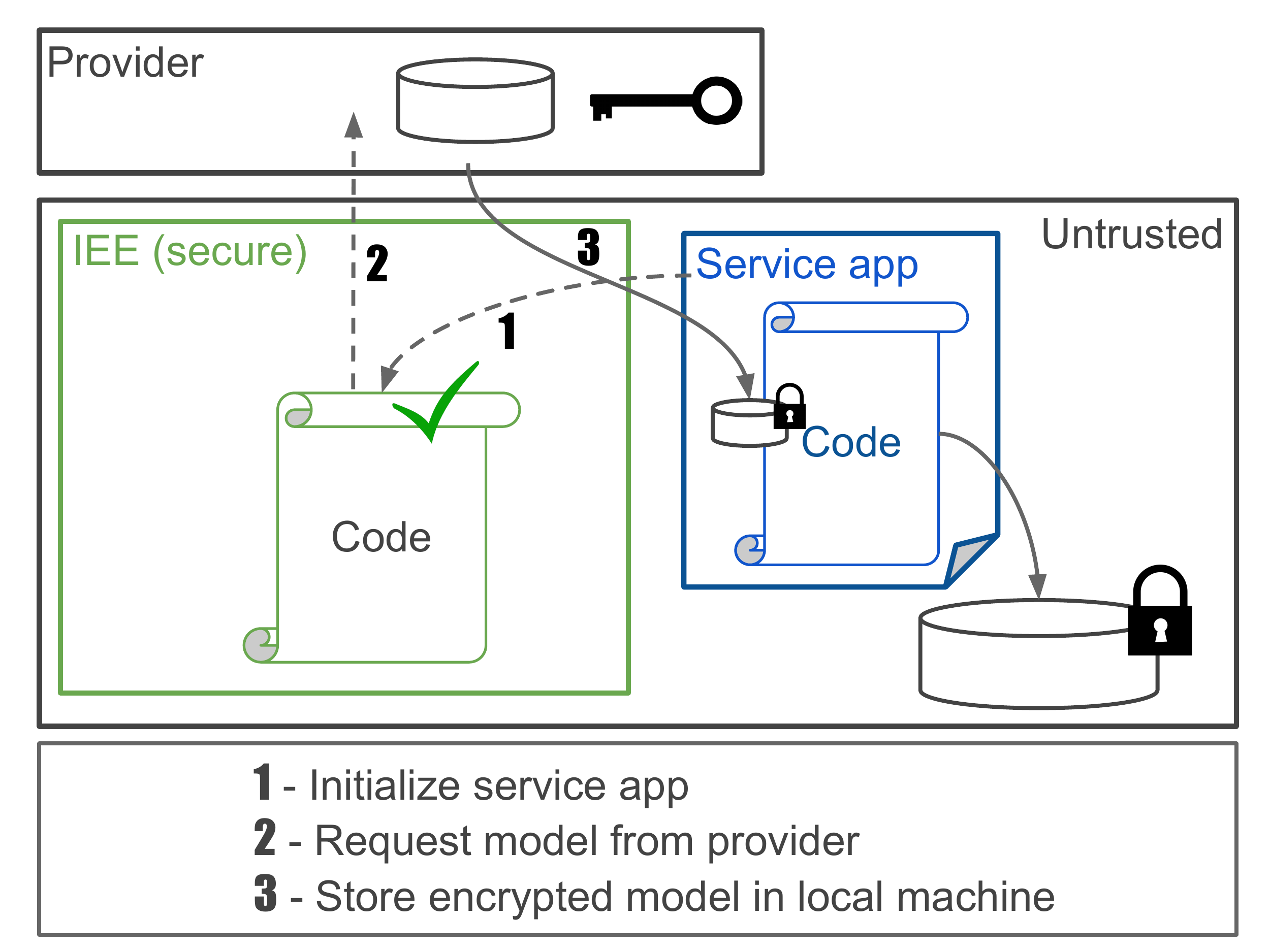}
\caption{Our scheme with all steps of initialization.}\label{fig:scheme}
\end{figure}
Once the setup is done, the client can use the enclave for the classification task.
This step is called the \emph{inference phase} and
is depicted in \autoref{fig:offlineClassification}.
\begin{figure} 
\includegraphics[width=0.97\columnwidth]{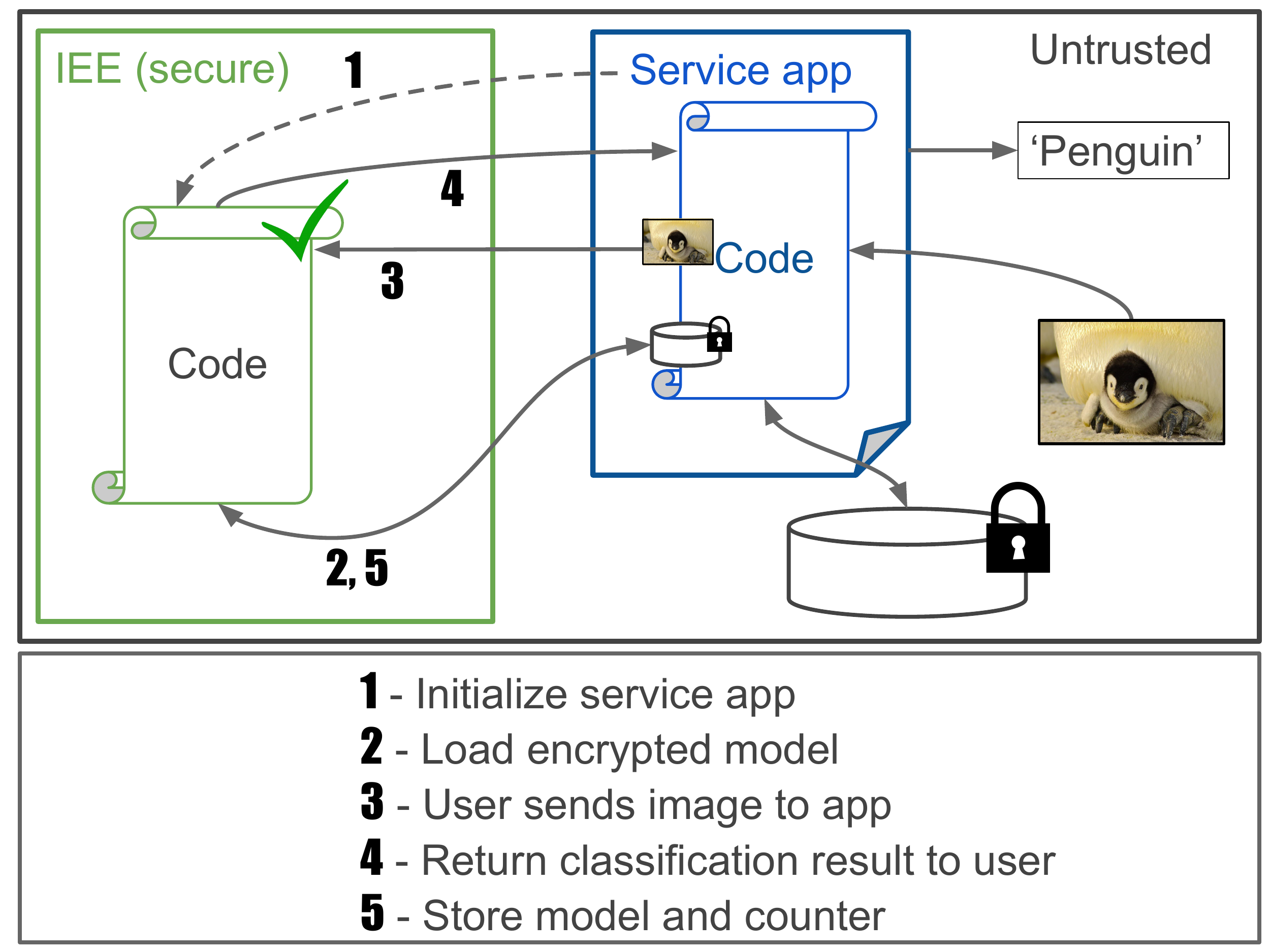}
\caption{Our scheme with all steps of offline classification.}\label{fig:offlineClassification}
\end{figure}
The isolation of the enclave ensures that the user is not able to infer more information about the model than given API access to a server-side model. 
Next, we describe \system in more details.

\medskip
\noindent\emph{Setup Phase.} The user and the service provider interact to set up a working enclave on the users platform. The user has to download the enclave's code and the code of a service application that will set up the enclave on the user's side. The enclave's code is provided by the SP, where the code of the service app can be provided by the SP but it can also be implemented by the user or any third party that she trusts. Note that the user can freely inspect the enclave's code and check what data is exchanged with the SP. This inspection ensures that the user's input to the ML model will be hidden from the SP, as the classification is run locally and the input data never leaves the user's platform.

After the user's platform attests that an enclave instance is running, the SP provides the enclave with secret data composed of among others the weights of the network. 
Finally, the enclave seals all secrets provided by the service provider and the service application stores it for further use. Sealing hides this data for other processes on the user's platform. With the end of the setup phase, no further communication with the SP is needed. 

\medskip
\noindent\emph{Inference Phase.} 
To perform classification, the user executes the service app and provides the test data as input. The service app restores the enclave, which can now be used to perform classification. Since the enclave requires the model parameters, the service app loads the sealed data stored during the setup phase. Before classification, the enclave can also perform an access control procedure that is based on the user's input data (available to the enclave in plaintext) and the current state of the enclave. 
Due to some limitations (e.g. limited memory of the IEE), the enclave can be implemented in a way that classification is performed layer wise, i.e. the service app provides sealed data for each layer of the network separately.
In the end, the enclave outputs the result to the service app and might as well update its state, which is stored inside the sealed data. This process is depicted in \autoref{fig:offlineClassification}.

\subsection{Discussion on Requirements}
Next, we discuss how \system fulfills the requirements stated in \autoref{sec:requirement}.

\medskip
\noindent\textbf{User Input Privacy.}
\system is executed locally by the user. Moreover, the user is allowed to inspect the code executed in the secure hardware. This means that she can check for any communication command with the SP and stop execution of the program. 
Moreover, off-line local execution ensures that the user's input data is private because there is no communication required with the SP. We conclude that \system perfectly protects the user input. 

\medskip
\noindent\textbf{Pay-per-query.}
To enforce the pay-per-query paradigm, the enclave will be set up during provision with a threshold. Moreover, the enclave will store a counter that is increased with every classification performed. Before performing classification, it is checked whether the counter exceeds this threshold. In case it does, the enclave returns an error instead of the classification. Otherwise, the enclave works normally.

This solution ensures that the user cannot exceed the threshold, which
means that she can only query the model for the number of times she paid. Unfortunately, it does not allow for a fine-grained pay-per-query, where the user can freely chose if she wants more queries at a given time. On the other hand, this is also the model currently used by server-side MLaaS, where a user pays for a fixed number of queries (e.g. 1000 in case of Google's vision API). \system can also be configured in a way that to perform offline inference the user has to receive a signature of the service provider under some function of the queried data, where this value is send in the clear by the user. 
The function has be to collision-resistant, so that the user cannot use the signature for two or more queries, and 
lossy, so that this value loses most of the information of the preimage. A practical candidate are cryptographic hash functions.
This configuration allow for a more fine-grained access to model but at the same time provides similar security guarantees to the
offline mode.

\medskip
\noindent\textbf{Intellectual Property.}
\system protects the service provider's intellectual property by ensuring that the isolation provided by the hardware simulates in a way the black-box access in the standard server-side model of MLaaS. In other words, the user gains no additional advantage in stealing the intellectual property in comparison to an access to the model through an server-side API. 
\section{Security Analysis}
\label{sec:secana}

In this section, we introduce a formal model for \system and show a concrete instantiation using the abstracted hardware model by Fisch et al. \cite{FVBG17} and a standard public key encryption scheme, which we recalled in \autoref{sec:back}. The goal is to prove that our construction possesses a property that we call ML model secrecy. 
We define it as a game that is played between the challenger and the adversary. The challenger chooses a bit which indicates whether the adversary is interacting with the real system or with a simulator. This simulator gets as input the same information as the adversary. In a nutshell, this means that the user can simulate \system using a server-side API ML model. Hence, classification using an IEE should not give the user more information about the model than using an oracle access to the model (e.g. as it is the case for a server-side API access). We would like to stress that this high-level view on security is currently the best we can provide. Any definition that tries to somehow quantify security of the model by bounding the leakage of model queries has to be based on empirical data. Thus, we circumvent this by stating that \system does not leak more information about the model than a model accessed by a server-side API.

We begin the description of the model by binding it with the high level description presented in \autoref{sec:sysdesign} using three algorithms that constitute the
interactive setup phase ($\Train$, $\Obtain$ and $\Provide$) and the local inference phase ($\Classify$).  
An overview of how \system fits into this composition and details about the instantiation are given in \autoref{fig:formal}.
Moreover, in the next subsection we describe the security model in more details and then present our formal instantiation of \system.

\begin{figure*}[t]
\centering
\includegraphics[width=1.8\columnwidth]{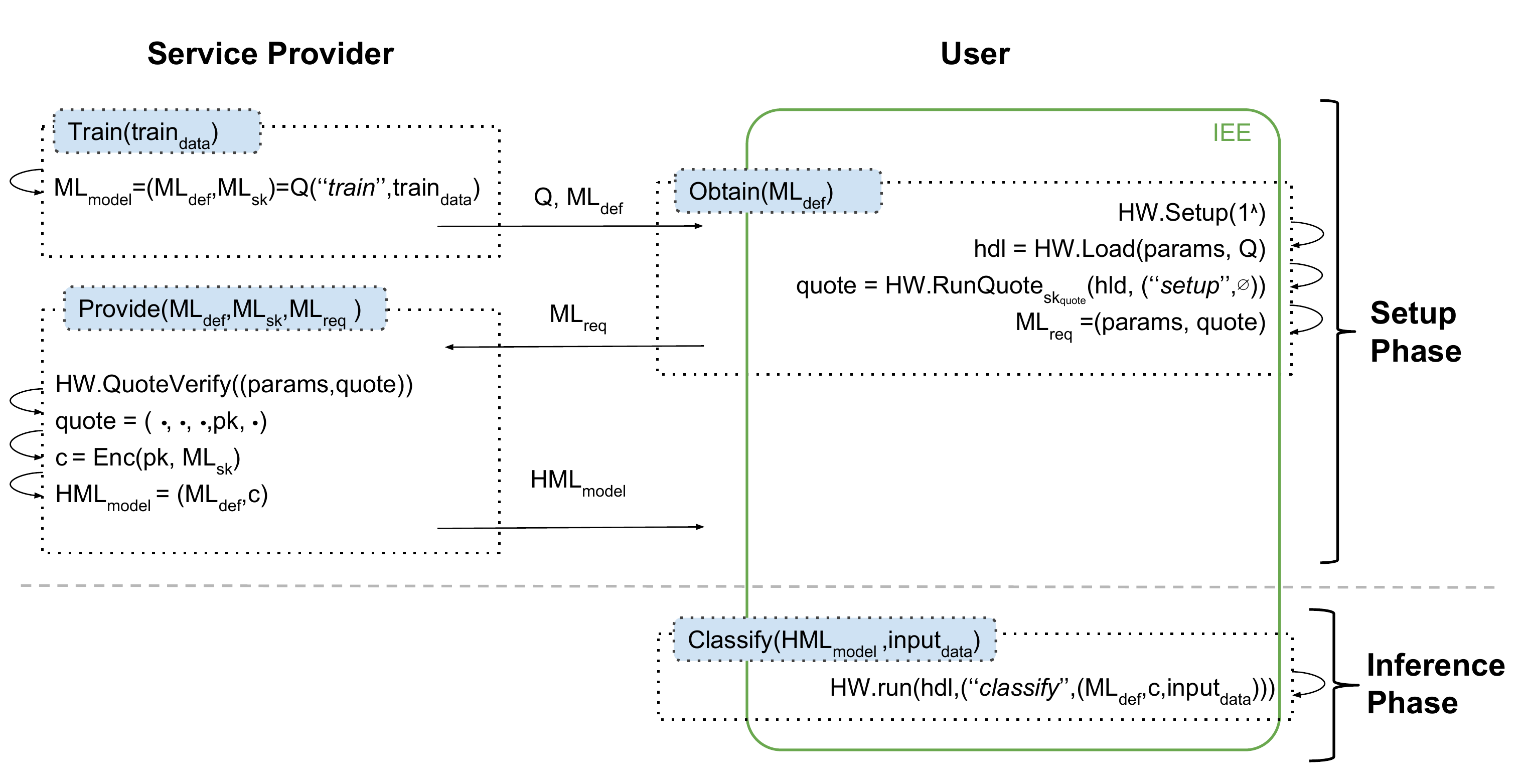}
\caption{Formal instantiation of $\system$.}\label{fig:formal}
\end{figure*}

\subsection{Formalization of \system}
We now define in more detail the inputs and outputs of the four algorithms that constitute the model for \system. Afterwards, we show what it means that \system is correct and show a game based definition of ML model secrecy.

\begin{itemize}
\item $\Train(\tdata)$: this probabilistic algorithm is executed by the service provider in order to create a machine learning model $\MLmodel = (\MLdef,\MLsk)$ based on training data $\tdata$, where $\MLdef$ is the definition of the ML model and $\MLsk$ are the secret weights and biases. To obtain the classification $\OutData$ for a given input data $\InData$ one can execute $\OutData = \MLmodel(\InData)$. 

\item $\Obtain(\MLdef)$: this algorithm is executed by the user to create a request $\MLreq$ to use model with definition $\MLdef$. Further, this algorithm is part of the setup phase and shown in steps 1 and 2 in \autoref{fig:scheme}.

\item $\Provide(\MLdef,\MLsk,\MLreq)$: this probabilistic algorithm is executed by the service provider to create a hidden ML model $\HMLmodel$ based on the request $\MLreq$.
\autoref{fig:scheme} depicts this algorithm in step 3 of the setup phase.

\item $\Classify(\HMLmodel,\InData)$: this algorithm is executed by the user to receive the output $\OutData$ of the classification. Hence, it models the inference phase depicted in \autoref{fig:offlineClassification}.
\end{itemize}

\medskip
\noindent\textbf{Correctness.} We say that \system is correct if for all training data $\tdata$, all ML models $(\MLdef,\MLsk) = \Train(\tdata)$,
all input data $\InData$ and all requests $\MLreq = \Obtain(\MLdef)$ we have
$\MLmodel(\InData) = \Classify(\HMLmodel,\InData)$, where 
$\HMLmodel = \Provide(\MLdef,\MLsk,\MLreq)$.

\medskip
\noindent\textbf{ML Model Secrecy.} We define model secrecy as a game played between a challenger $C$ and an adversary $\Adv$. Depending on the bit chosen by the 
challenger, the adversary interacts with the real system or a simulation. More formally, we say that \system is ML model secure if there exists a simulator $\Sim = (\Sim_1,\Sim_2)$ such that the probability that $\lvert \Pr[\ExpModelSec{-0}(\lambda) = 1] - \Pr[\ExpModelSec{-1}(\lambda) = 1] \rvert$ is negligible for any probabilistic polynomial time adversary $\Adv$.

\subsection{Instantiation of \system} \label{subsec:system}

The idea behind our instantiation is as follows. The user retrieves a program $Q$ from the SP and executes it inside a secure hardware.
This secure hardware outputs a public key and an attestation that the code was correctly executed. This data is then send to the SP, which encrypts the secrets corresponding to the ML model and sends it back to the user. This ciphertext is actually the hidden machine learning model, which is decrypted by the hardware and the plaintext is used inside the hardware for classification. ML model secrecy follows from the fact that the user cannot produce forged attestations without running program $Q$ in isolation. This also means that the public key is generated by the hardware and due to indistinguishability of chosen plaintext of the encryption scheme, we can replace the ciphertext with an encryption of $0$ and answer hardware calls using the model directly and not the $\Classify$ algorithm. Below we present this idea in more details.

\begin{algorithm}[t]
 \caption{Hardware Program $Q$}
 \begin{algorithmic}[1]
 \renewcommand{\algorithmicrequire}{\textbf{Input:}}
 \renewcommand{\algorithmicensure}{\textbf{Output:}}
 \REQUIRE \command, \data
 \ENSURE  \out
  \IF {\command == "train"}
  \STATE run the ML training on $\data$ receiving $\MLmodel = (\MLdef,\MLsk)$,
  store $\MLmodel$ and set $\out = \MLmodel$
  \ELSIF {\command == "setup"}
  \STATE execute $(\sk,\pk) = \KeyGen(\secparam)$, store $\sk$ and set $\out = \pk$
  \ELSIF {\command == "classify"}
  \STATE parse $\data = (\MLdef,c,\InData)$ and execute $\Dec(\sk,c) = \MLsk$, set $\MLmodel = (\MLdef,\MLsk)$ and $\out = \MLmodel(\InData)$
  \ENDIF
 \RETURN $\out$
 \end{algorithmic}
 \end{algorithm}

\begin{itemize}
\item $\Train(\tdata)$: executes $\MLmodel = (\MLdef,\MLsk)=$ \linebreak $Q("train",\tdata)$.
Output $\MLdef$.

\item $\Obtain(\MLdef)$: given $\MLdef$, set up the hardware parameters $\params$ by running $\HWSetup(\secparam)$. Load program $Q$ using $\HWLoad(\params,Q)$, further receive a handle to the enclave $\hdl$. 
Execute the $\HW$ setup command for program $Q$ by running $\HWRunQuote(\hdl,("setup",\emptyset))$ and receive a $\HWquote$ (that includes the public key $\pk$).
Finally, set $\MLreq = (\params,\HWquote)$.

\item $\Provide(\MLdef,\MLsk,\MLreq)$:
abort if the quote verification failed: $\HWQuoteVerify(\params, \allowbreak \HWquote) = 0$. 
Parse $\HWquote = (\cdot,\cdot,\cdot,\pk,\cdot)$, 
compute ciphertext 
$c = \Enc(\pk,\MLsk)$ and set $\HMLmodel = (\MLdef,c)$.

\item $\Classify(\HMLmodel,\InData)$: 
parse the hidden model \linebreak $\HMLmodel$ as  $(\MLdef,c)$ and return the output of $$\HWRun(\hdl,("classify",(\MLdef,c,\InData))).$$

\end{itemize}

\begin{figure}
  \begin{pcvstack}[center]
    \procedure{$\ExpModelSec{-b}(\lambda)$}{%
      \MLmodel = (\MLdef,\MLsk) \exec \Train(\tdata) \\
      \MLreq \exec \Adv(\MLdef) \\
      \pcif b=0 \; \HMLmodel \exec \Sim_1(\MLdef,\MLreq) \\
      \pcelse \HMLmodel \exec \Provide(\MLdef,\MLsk,\MLreq)\\
      \pcreturn \hat{b} \exec \Adv^{\Or(\MLmodel,\HMLmodel,\cdot)}(\HMLmodel) \\
    }
     \pcvspace
   \vspace{-0.7cm}
    \procedure{$\Or(\MLmodel,\HMLmodel,\InData)$}{%
      \text{parse } \MLmodel = (\MLdef,\cdot)\\
      \pcif b=0 \; \pcreturn \Sim_2(\MLmodel,\InData) \\
      \pcelse \pcreturn \Classify(\MLdef,\HMLmodel,\InData)
    }
  \end{pcvstack}
  \caption{ML model secrecy experiment.} \label{fig:mlsecrecy}
\end{figure}

\subsection{Security}

\begin{theorem}
The \system presented in \autoref{subsec:system} is model secure if in the used public key encryption is indistinguishable under chosen plaintext attacks and the hardware functionality $\HW$
is remote attestation unforgeable.
\end{theorem}
\begin{proof}
We prove this theorem using the game based approach starting with $\mathbf{GAME_0}$ that models the original ML model security experiment with bit $b=1$ and end the proof with $\mathbf{GAME_4}$, which is the experiment for $b=0$.

\medskip
\noindent $\mathbf{GAME_1}$. Similar to $\mathbf{GAME_0}$, but we abort in case the adversary outputs a valid request $\MLreq$ without running program $Q$.
It is easy to see that by making this change, we only lower the adversary's advantage by a negligible factor. In particular, an adversary for which we abort $\mathbf{GAME_1}$ can be used to break the remote attestation unforgeability of the hardware model. 

\medskip
\noindent $\mathbf{GAME_2}$. We now replace the way  oracle $\Or$ works. On a given query of the adversary,
we always run $\Sim_2(\MLmodel,\InData) = \MLmodel(\InData)$.
Note that this does not change the adversary's advantage, since both outputs should by correctness of \system give the same output on the same input.

\medskip
\noindent $\mathbf{GAME_3}$. We now replace the ciphertext given as part of $\HMLmodel = (\MLdef,c)$ to the adversary, i.e. we replace $c$ with an encryption of $0$.
This change only lowers the adversary's advantage by the advantage of breaking the security of the used encryption scheme. Note that by $\mathbf{GAME_1}$ we ensured that $\pk$ inside the request $\MLreq$ is chosen by the secure hardware and can be set by the reduction. Moreover, oracle $\Or$ works independently of $\HMLmodel$.

\medskip
\noindent $\mathbf{GAME_4}$. We change how $\HMLmodel$ is computed. Instead of running 
$\Provide(\MLdef,\MLsk,\MLreq)$, we use 
$\Sim_2(\MLdef,\MLreq) = (\MLdef, \Enc(\pk,0))$, where $\MLreq = (\cdot,(\cdot,\cdot,\cdot,\pk,\cdot))$.
It is easy to see that this game is actually the experiment $\ExpModelSec{-0}(\lambda)$ and we can conclude that our instantiation is ML model secure because the difference between experiments $\ExpModelSec{-0}(\lambda)$ and $\ExpModelSec{-1}(\lambda)$ is negligible.
\end{proof}
\section{SGX Implementation and Evaluation}
\label{sec:implement}
In the setup phase, \system attests the execution of the enclave and decrypts the data send by the service provider. Both tasks are standard and supported by Intel's crypto library~\cite{AM16}. Thus, in the evaluation we focus on the inference phase and overhead the IEE introduces to classification.

\begin{table*}[!t]
	\centering
	\caption{Average dense layer overhead for 100 executions. This comparison includes two ways of compiling the code, i.e. with and without the g++ optimization parameter -O3.}\label{table:dense_layer}
    \scalebox{1}{
	\begin{tabular}{l r c r r}
	    \toprule
	 	Matrix Dimension & \system Layer (no -O3) & \system Layer & Standard Layer (no -O3) & Standard Layer  \\
        \midrule  
        256$\times$256 & $0.401$ms & $0.234$ms & $0.164$ms & $0.020$ms\\
        512$\times$512 & $1.521$ms & $0.865$ms & $0.637$ms  & $0.062$ms \\
        1024$\times$1024 & $6.596$ms & $4.035$ms & $2.522$ms & $0.244$ms\\
        2048$\times$2048 & $37.107$ms & $26.940$ms & $10.155$ms & $1.090$ms\\
        4096$\times$4096 & $128.390$ms & $96.823$ms & $40.773$ms & $4.648$ms\\
		\bottomrule
	\end{tabular}
    }
\end{table*}
\begin{table}[!t]
	\centering
	\caption{Average convolution layer (above) and depthwise separable convolution layer (below) overhead for 100 executions and $3\times 3$ filters.}\label{table:overheadconv2d}
\setlength{\tabcolsep}{3pt}
	\begin{tabular}{l r r r}
	    \toprule
	 	Input/Output Size & \system Layer & Standard Layer& Factor \\
        \midrule  
        64$\times$224$\times$224 & $80$ms & $66$ms & $1.21$\\
        512$\times$28$\times$28 & $61$ms & $51$ms & $1.20$\\
        512$\times$14$\times$14 & $30$ms & $13$ms & $2.31$\\
        \hline
        64$\times$224$\times$224 & $41$ms & $27$ms & $1.52$\\
        512$\times$28$\times$28 & $7$ms & $7$ms & $1.00$\\
        512$\times$14$\times$14 & $2$ms & $2$ms & $1.00$\\
		\bottomrule
	\end{tabular}
\end{table}

\begin{table}[!t]
	\centering
	\caption{Average neural network evaluation overhead for 100 executions.}\label{table:networks}
    \scalebox{1}{
	\begin{tabular}{l r r r}
	    \toprule
	 	Network & \system Layer & Standard Layer & Factor \\
        \midrule  
        VGG-16 & $1145$ms & $736$ms & $1.55$\\
        MobileNet & $427$ms & $197$ms & $2.16$ \\
		\bottomrule
	\end{tabular}
    }
\end{table}

\subsection{Implementation}
We used an Intel SGX enabled desktop PC with Intel(R) Core(TM) i7-6700 CPU @ 3.40GHz that was running Ubuntu 18.04. The implementation was done using C++ and based on the code of Slalom \cite{TB18}, which uses a custom lightweight C++ library for feed-forward networks based on Eigen. If not stated otherwise, we used the \mbox{-O3} compiler optimization and C++11 standard.
Yet, porting well-known ML frameworks (for example TensorFlow) to SGX is not feasible at this point, because enclave code cannot access OS-provided features (e.g. multi-threading, disk, and driver IO). However, any model used in those frameworks can be
easily translated into \system.

We wrap standard layers to create new \system layers. Those layers take the standard input of the model layer but the weights are given in sealed form. Inside the enclave, the secret data is unsealed and forwarded to the standard layer function. \system layers are designed to be executed inside the enclave by providing \verb+ECALL+'s. See \autoref{fig:systemlayer} for more details. This approach provides means to build \system secure neural networks in a modular way.

\begin{figure}[!t]
\centering
\includegraphics[width=\columnwidth]{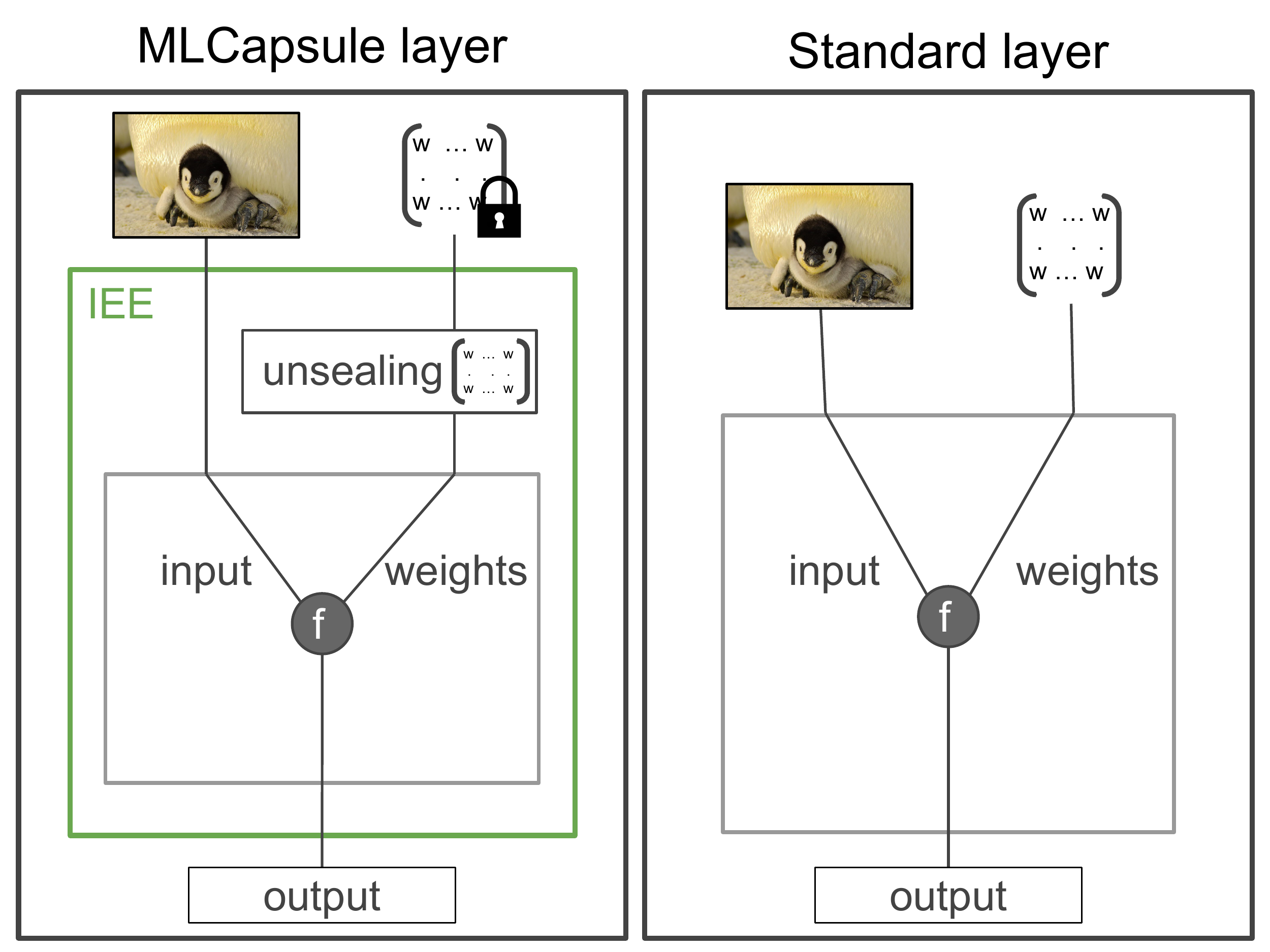}
\caption{Difference between \system and standard layer.}\label{fig:systemlayer}
\end{figure}

Since the sealed data is provided from outside the enclave, it has to be copied to the enclave before unsealing. Otherwise, unsealing will fail. We measure the execution time of \system layers as the time it takes to:
($1$) allocate all required local memory, ($2$) copy the sealed data to the inside of the enclave, ($3$) unseal the data, ($4$) perform the standard computation using the unsealed weights and plaintext input and finally ($5$) free the allocated memory.

\medskip
\noindent\textbf{Implementation Issues.}
Applications are limited to 90 MB of memory, because there is currently no SGX support for memory swapping. Linux provides an OS based memory swap, but the enclave size has to be final and should not expand to prevent page faults and degraded performance. This performance issue is visible in our results for a dense layer with weight matrix of size $4096 \times 4096$. In this case, the \system layer allocates $4 \times 4096 \times 4096 = 64$ MB for the matrix and a temporary array of the same size for the sealed data. Thus, we exceed the 90 MB limit, which leads to a decrease in performance. In particular, the execution of such a layer took $1$s and after optimization,
the execution time decreased to $0.097$s.

We overcome this problem by encrypting the data in chunks of 2 MB. 
This way the only large data array allocated inside the \system layer is the memory of the weight matrix. Using 2 MB chunks the \system layers requires only around $2 \times 2$ MB more memory than implementations of standard ML layers. 
We implemented this optimization only for data that requires more than 48 MB, e.g. in case of a VGG-16 network we used it for the first and second dense layer.

\subsection{Evaluation}

\noindent\textbf{Evaluation of Layers.}
Comparisons between \system and standard layers are given in \autoref{table:dense_layer} and \autoref{table:overheadconv2d}. 
It is worth noting that those results can be used to approximate the overhead of \system for arbitrarily networks. However, below we evaluate two full networks for image recognition. We do not compare directly to state-of-the-art cryptographic 
solutions because they use a server side ML model architecture. On the other hand, we compare
the direct execution of the computation on the CPU and inside SGX, as those results provide the overhead of implementing \system.

From \autoref{table:overheadconv2d}, we see that the overhead for convolutional layers averages around 1.2, with a peak to 2.3 with inputs of size $512\times 14\times 14$. 
In case of depth-wise separable convolutional layers, the execution time of \system layers is comparable with the execution time of standard layers. In this case, the difference is almost not noticeable for smaller input sizes. 
Applying additional activation functions or/and pooling after the convolution does not significantly influence the execution time. In case of dense layers, we observe a larger overhead. For all the kernel dimension the overhead is not larger than 25 times. We also evaluate dense layers without -O3 optimization. The results show that in such a case the overhead of \system is around the factor 3. We suspect that the compiler is able to more efficiently optimize the source code that does not use SGX specific library calls. Hence, the increase in performance is due to the optimized compilation. 

\medskip
\noindent\textbf{Evaluation of Full Classification.}
In this subsection we combine \system layers to form popular ML networks, i.e. VGG-16 and MobileNet. The first network can be used to classify images and work with the ImageNet dataset (size of images $224 \times 224 \times 3$). Similar, the second network can also be used for the same task. It is easy to see from \autoref{table:networks} that \system has around $2$-times overhead in comparison to an execution of the classification using standard layer and without the protection of SGX.
\section{Advanced Defenses} 
\label{sec:advdef}

Recently, researchers have proposed multiple attacks against MLaaS: 
reverse engineering~\cite{OASF18}, 
model stealing~\cite{TZJRR16,WG18}, 
and membership inference~\cite{SSSS17,SZHBFB19}. 
As mentioned in \autoref{sec:requirement}, 
these attacks only require a black-box access (API) to the target ML model,
therefore, their attack surface is orthogonal 
to the one caused by providing the model's white-box access to an adversary.
As shown in the literature, 
real-world MLaaS suffers from these attacks~\cite{TZJRR16,SZHBFB19}. 

In this section, we propose two new defense mechanisms against reverse engineering and membership inference (test-time defense).
We show that these mechanisms
together with a proposed defense for model stealing
can be seamlessly integrated into \system. 

\subsection{Detecting Reverse Engineering}
Oh et al.~have shown that by only having black box access to neural network model, a wide variety of model specifics can be inferred \cite{OASF18}. Such information includes training procedure, type of non-linearities, and filter sizes, which thereby turns the model step by step into a white-box model. This equally affects the safety of intellectual property and increases attack surface. 
	
\medskip
\noindent\textbf{Methodology.}
Up to now, no defense has been proposed to counter this attack. We propose the first defense in this domain and show that it can be implemented seamlessly into \system. We observe that the most effective method proposed by Oh et al.~\cite{OASF18} relies on crafted input patterns that are distinct from benign input. Therefore, we propose to train a classifier that detects such malicious inputs. Once a malicious input is recognized, service can be immediately denied which stops the model from leaking further information. Note that also this detection is running on the client and therefore the decision to deny service can be taken on the client and does not require interaction with a server.

We focus on the kennen-io method by Oh et al.~\cite{OASF18} as it leads to the strongest attack. We also duplicate the test setup on the MNIST dataset.\footnote{\url{http://yann.lecun.com/exdb/mnist/}} 
In order to train a classifier to detect such malicious inputs, we generate 4,500 crafted input image with the kennen-io method and train a classifier against 4,500 benign MNIST images. We use the following deep learning architecture: 
\begin{align*}
\texttt{Input Image} \rightarrow
\texttt{conv2d($5\times 5$, 10)} & \\
\texttt{max($2\times 2$)} & \\
\texttt{conv2d($5 \times 5$, 20)} & \\
\texttt{max($2\times 2$)} & \\
\texttt{FullyConnected($50$)} & \\
\texttt{FullyConnected($2$)} & \\
\texttt{softmax}  & \rightarrow \texttt{Output}
\end{align*}
where \texttt{conv2d($a\times a$, b)} denotes a 2d convolution with $a$ by $a$ filter kernel and $b$ filters, \texttt{max($c\times c$)} denotes max-pooling with a kernels size of $c$ by $c$, \texttt{FullyConnected($d$)} denotes a fully connected layer with $d$ hidden units and \texttt{softmax} a softmax layer. 
The network uses ReLU non-linearties and drop-out for regularization. We represent the output as 2 units -- one for malicious and one for benign. We use a cross-entropy loss to train the network with the ADAM optimizer. 

\begin{table*}[!t]
	\centering
	\caption{Overhead of detecting a model stealing attack, Prada. We assume $3000$ samples in the detection set, enough to query $5000$ benign samples. }\label{table:prada}
	\begin{tabular}{l r r r r}
	    \toprule
	 	Dataset & Size & SGX & Outside SGX & Factor \\
        \midrule  
        MNIST & 1$\times$28$\times$28 & $35$ms & $2.6$ms & $13.5$\\
        CIFAR & 3$\times$32$\times$32 & $38$ms & $10.1$ms & $3.8$ \\
        GTSRB & 3$\times$215$\times$215 & $2200$ms & $440$ms & $5$\\
		\bottomrule
	\end{tabular}
\end{table*}

\medskip
\noindent\textbf{Evaluation.}
We compose a test set of additional 500 malicious input samples and 500 benign MNIST samples that are disjoint from the training set.
The accuracy of this classifier is $100\%$, thereby detecting each attack on the first malicious example, which in turn can be stopped immediately by denying the service. 
Meanwhile, no benign sample leads to a denied service. This is a very effective protection mechanism that seamlessly integrates into our deployment model and only adds 0.832 ms (per image classification) to the overall computation time. While we are able to show very strong performance on this MNIST setup, it has to be noted that the kennen-io method is not designed to be ``stealthy'' and future improvements of the attack can be conceived that make detection substantially more difficult.

\subsection{Detecting Model Stealing}
Model stealing attacks aim at obtaining a copy from an MLaaS model~\cite{TZJRR16,PMG16}. 
Usually, this is done by training a substitute on samples rated by the victim model, resulting in a model with similar behavior and/or accuracy.
Hence, successful model stealing leads to the direct violation of the service provider's intellectual property.
Very recently, Juuti el al.~\cite{JSDMA18} propose a defense,
namely Prada,
to mitigate this attack
which we implement in~\system.

Prada maintains a growing set of user-submitted queries. Whether a query is appended to this growing set depends on the minimum distance to previous queries and a user set threshold. Benign queries lead to a constant growing set, whereas Juuti et al.\ show that malicious samples generally do not increase set size. Hence, an attack can be detected by the difference in the growth of those sets.  

As the detection is independent of the classifier, it can be easily implemented in~\system. 
The resulting computation overhead depends heavily on the user submitted queries~\cite{JSDMA18}. We thus measure the general overhead of first loading the data in the enclave and second of further computations. 

Juuti et al.\ state that the data needed per client is $1$-$20$MB. We plot the overhead with respect to the data size in~\autoref{fig:prada}. We observe that the overhead for less than $45$MB is below 0.1s.
Afterwards, there is a sharp increase, as the heap size of the enclave is $90$MB: storing more data requires several additional memory operations.   
For each new query, we further compute its distance to all previous queries in the set---a costly operation.
We assume a set size of $3,000$, corresponding to roughly $5,000$ benign queries.
\autoref{table:prada} shows that a query on the GTSRB dataset\footnote{\url{http://benchmark.ini.rub.de/}} 
is delayed by almost 2s, or a factor of five. For datasets with smaller samples such as CIFAR~\footnote{\url{https://www.cs.toronto.edu/~kriz/cifar.html.}}
or MNIST, the delay is around 35ms. 

\begin{figure}[!t]
\centering
\includegraphics[width=\linewidth]{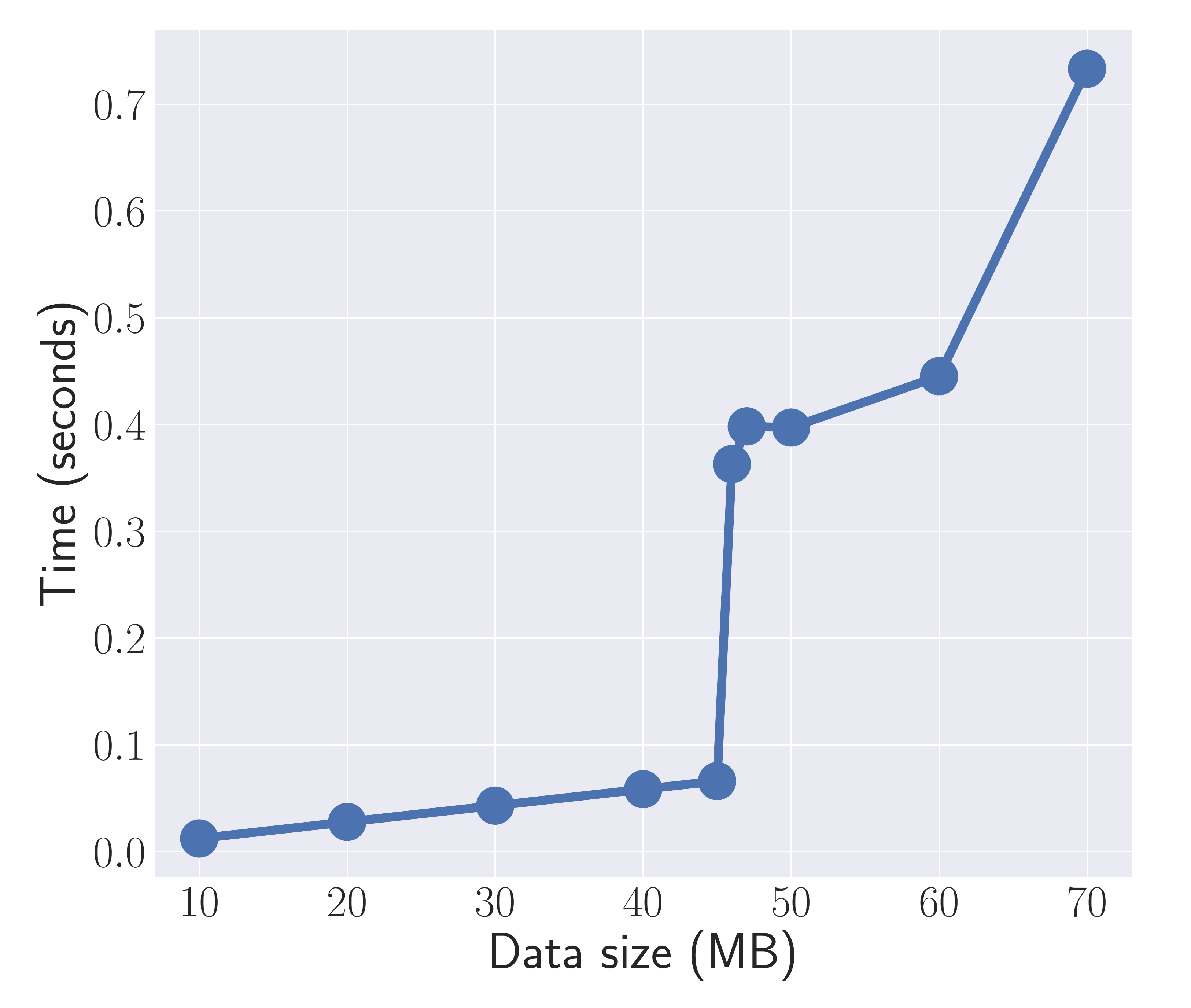}
\caption{Overhead in seconds to load additional data for a defense mechanism preventing model stealing, namely Prada~\cite{JSDMA18}.}\label{fig:prada}
\end{figure}

\begin{figure*}[!t]
\centering
\begin{subfigure}{\columnwidth}
\includegraphics[width=\columnwidth]{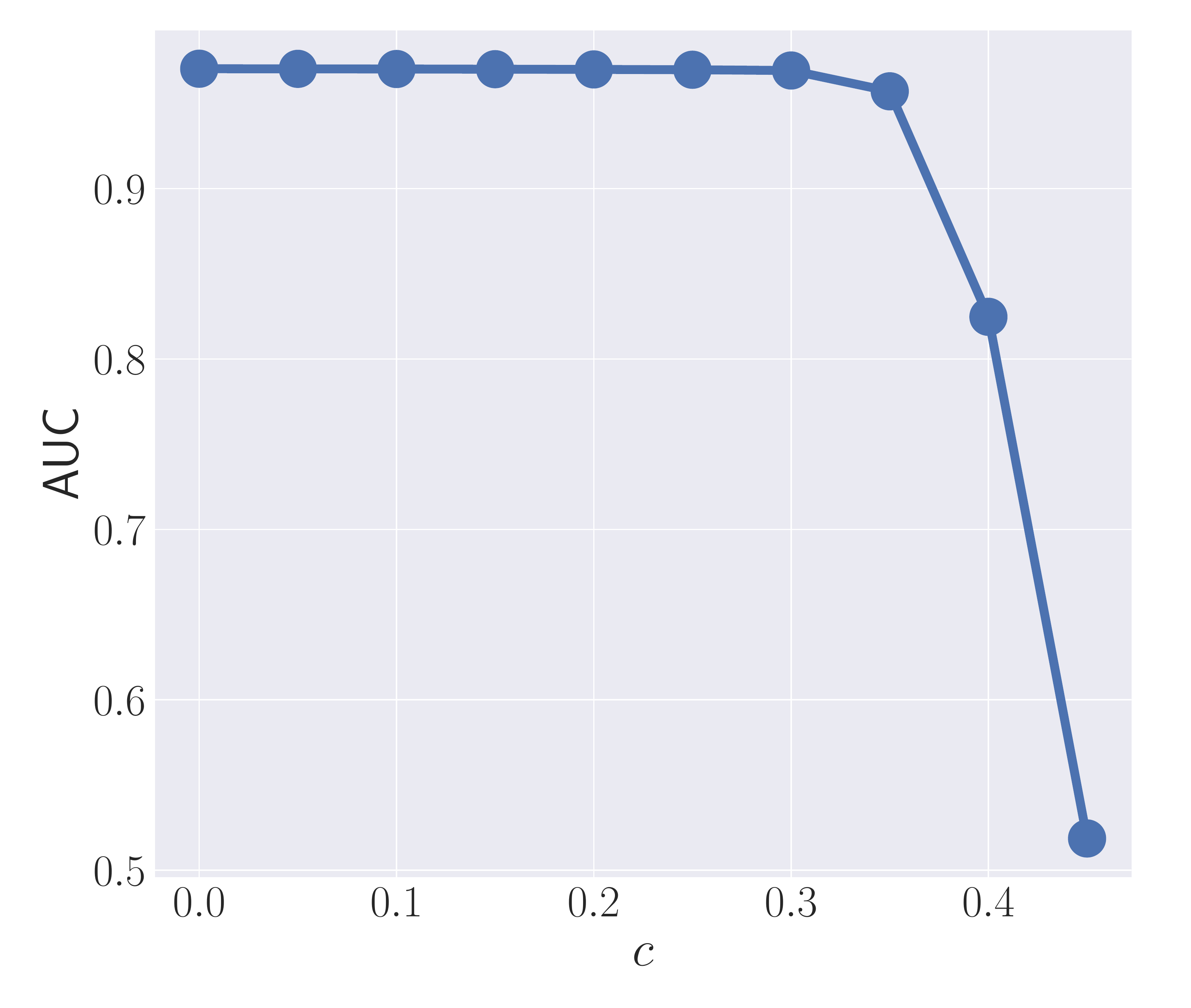}
\caption{}\label{fig:meminf_attack}
\end{subfigure}
\begin{subfigure}{\columnwidth}
\includegraphics[width=\columnwidth]{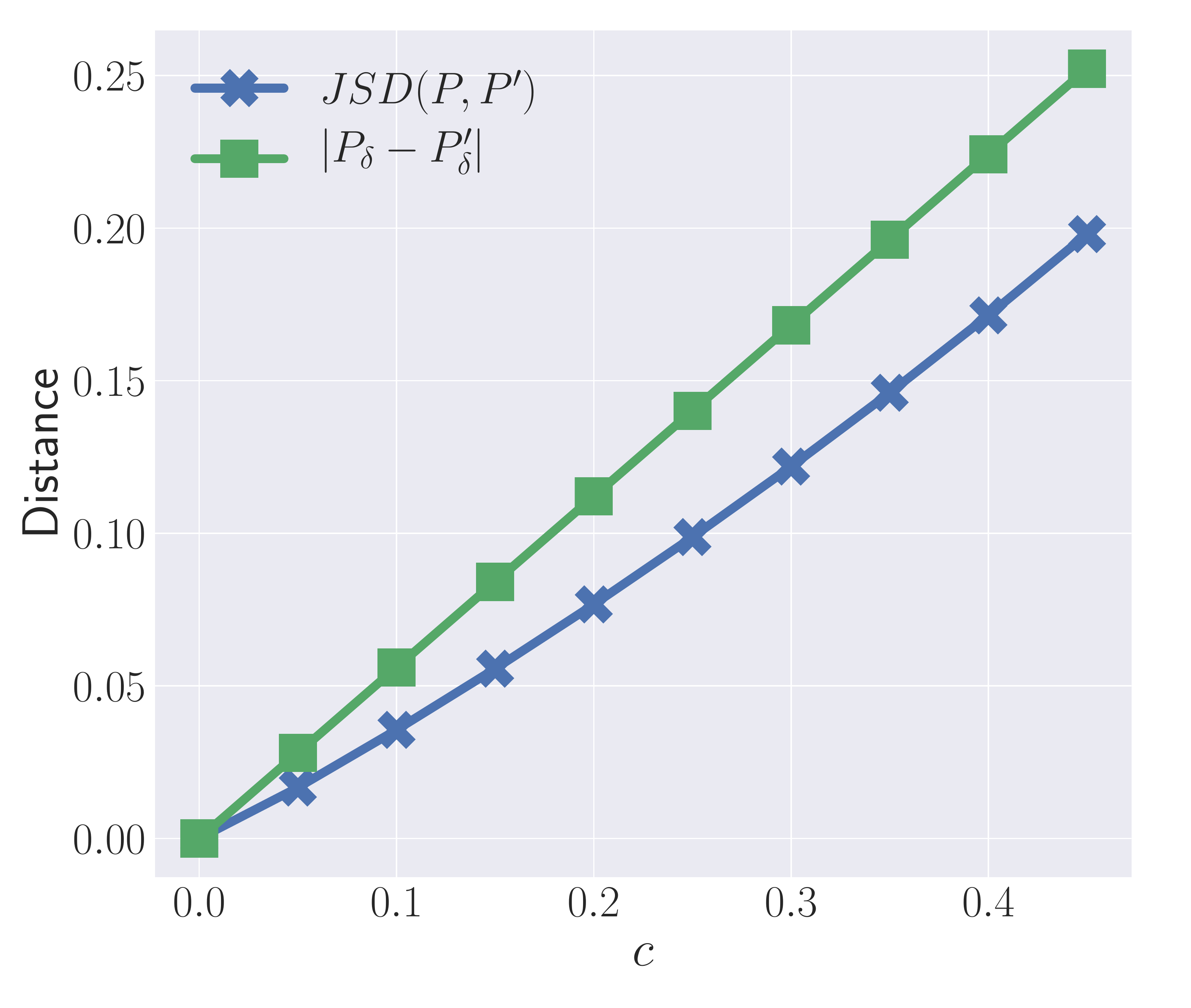}
\caption{}\label{fig:meminf_utlity}
\end{subfigure}
\caption{The relation between 
the hyperparameter controlling the noise magnitude, i.e., $c$, 
and (a) membership prediction performance 
and (b) target model utility.
${\it JSD}(\Post, \Post')$ denotes the Jensen-Shannon divergence
between the original posterior $\Post$ and the noised $\Post'$,
while $\vert \Post_\delta-\Post'_\delta \vert$ is the absolute difference between
the correct class's posterior ($\Post_\delta$)
and the noised one ($\Post'_\delta$).
}
\end{figure*}

\subsection{Membership Inference}

Shokri et al.\ demonstrate 
that ML models are vulnerable to membership inference~\cite{SSSS17}
due to overfitting:
A trained model is more confident 
facing a data point it was trained on than facing a new one (reflected in the model's posterior).
They propose to use a binary classifier to perform membership inference,
and rely on shadow models
to derive the data for training the classifier.
Even though the attack is effective, it is complicated and expensive to mount.
More recently, Salem et al.\ relax the assumption of the threat model by Shokri et al.\ 
and show that an adversary only needs the posterior's entropy to achieve a similar performance~\cite{SZHBFB19}.
This attack causes more severe privacy damages,
thus we use it in our evaluation.
However, we emphasize that our defense is general and can be applied to other membership inference attacks as well.

\begin{algorithm}[!t]
\caption{Noising mechanism to mitigate membership inference attack.}
\label{alg:meminf}
 \begin{algorithmic}[1]
 \renewcommand{\algorithmicrequire}{\textbf{Input:}}
 \renewcommand{\algorithmicensure}{\textbf{Output:}}
 \REQUIRE Posterior of a data point $\Post$, Noise posterior \Noise
 \ENSURE Noised posterior $\Post'$
\STATE Calculate $\entropy(\Post)$ \# the entropy of $\Post$\\
\STATE $\magnitude= 1-\frac{\entropy(\Post)}{\log\vert \Post \vert}$ \# the magnitude of the noise\\
\STATE $\Post' = (1-c\magnitude)\Post + c\magnitude T$
 \RETURN $\Post'$
\end{algorithmic}
\end{algorithm}

\medskip
\noindent\textbf{Methodology.}
We define the posterior of an ML model predicting a certain data point as a vector $\Post$, 
and each class $i$'s posterior is denoted by $\Post_i$.
The entropy of the posterior is defined as:
\[
\entropy(\Post) = -\sum_{\Post_i\in \Post} \Post_i \log \Post_i.
\]
Lower entropy implies the ML model is more confident on the corresponding data point.
Following Salem et al.,
the attacker predicts a data point with entropy smaller than a certain threshold 
as a member of the target model's training set, and vice versa~\cite{SZHBFB19}.

The principle of our defense is adding more (less) noise to a posterior with low (high) entropy, and publishing the noised posterior.
The method is listed in~\autoref{alg:meminf}.
In Step 1, we calculate $\entropy(\Post)$.
In Step 2, we derive from $\entropy(\Post)$ 
the magnitude of the noise, i.e., 
\[
\magnitude = 1-\frac{\entropy{\Post}}{\log\vert \Post \vert}.
\]
Here, $\frac{\entropy{\Post}}{\log\vert \Post \vert}$ is the normalized $\entropy(\Post)$
which lies in the range between 0 and 1.
Hence, lower entropy implies higher $\magnitude$, i.e., larger noise,
which implements the intuition of our defense.
However, according to our experiments,
directly using $\magnitude$ generates too much noise to $\Post$.
Thus, we introduce a hyperparameter, $c$, to control the magnitude $\magnitude$:
$c$ is in the range between 0 and 1,
its value is set following cross validation.
In Step 3, we add noise $\Noise$ to $P$ with $c\magnitude$ as the weight.
There are multiple ways to initialize $\Noise$, here, 
we define it as the class distribution of the training data.
Larger $c\magnitude$ will cause the final noised 
$\Post'$ to be more similar to the prior,
which reduces the information provided by the ML model.

Our defense is a test-time defense, i.e., it does not modify the original ML model.
Meanwhile, previous defense mechanisms modify the ML model
which may affect the model's performance~\cite{SSSS17,SZHBFB19}.
Moreover, previous mechanisms
treat all data points equally, 
even those that are unlikely to be in the training data.
In contrary, our defense mechanism adds different noise based on the entropy of the posterior.

\medskip
\noindent\textbf{Evaluation.}
For demonstration, we perform experiments with VGG-16 trained on the CIFAR-100 dataset.
Our experimental setup follows previous works~\cite{SZHBFB19}.
Mainly, we divide the dataset into two equal parts, and use one part to train the VGG model, i.e., the training set, the other part is referred to as the testing set.
Every data point in the training set is essentially a member for membership inference,
while every data point in the testing set is a non-member.

\autoref{fig:meminf_attack} shows the result.
As we can see, setting $c$ to $0$, i.e., not adding any noise, 
results in a high AUC score (0.97) 
which means the attack is very effective.
The AUC score starts dropping when increasing the value of $c$ as expected.
When the value of $c$ approaches $0.5$, the AUC score drops to almost 0.5, this means the best an attacker can do is random guessing the membership state of a point.

We also study the utility of our defense, 
i.e., how added noise affects the performance of the original ML model.
From \autoref{alg:meminf},
we see that our defense mechanism only adjusts the confidence values 
in a way that the predicted labels stay the same.
This means the target model's accuracy does not change.

To perform an in-depth and fair utility analysis,
we report the amount of noise added to the posterior.
Concretely, we measure the Jensen-Shannon divergence
between the original posterior ($\Post$) and the noised one ($\Post'$),
denoted by ${\it JSD}(\Post, \Post')$,
following previous works~\cite{MPS13,BHPZ17}.
Formally, ${\it JSD}(\Post, \Post')$ is defined as:
\[
{\it JSD}(\Post, \Post') = 
\sum_{\Post_i \in \Post}\Post_i \log\frac{\Post_i}{M_i}+
\Post'_i \log\frac{\Post'_i}{M_i}
\]
where $M_i = \frac{\Post_i+\Post'_i}{2}$.
Moreover, we measure the absolute difference
between the correct class's original posterior ($\Post_\delta$)
and its noised version ($\Post'_\delta$),
i.e., $\vert\Post_\delta-\Post'_\delta \vert$,
this is also referred to as the expected estimation error 
in the literature~\cite{STBH11,BHZLEB18,ZHRLPB18}.
In \autoref{fig:meminf_utlity},
we see that both ${\it JSD}(\Post, \Post')$ 
and $\vert \Post_\delta-\Post'_\delta \vert$
increase monotonically 
with the amount of noise being added (reflected by $c$).
However, when $c$ is approaching 0.5,
i.e., our defense mechanism can mitigate the membership inference risk completely,
${\it JSD}(\Post, \Post')$ 
and $\vert \Post_\delta-\Post'_\delta \vert$
are still both below 0.25: Our defense mechanism
is able to preserve the target model's utility to a large extent.

We measure the overhead of this defense
and it only adds 0.026ms to the whole computation.
This indicates our defense can be very well integrated into~\system.

\section{Discussion}
\label{sec:discussion}

Our formal proof shows that our setting is indistinguishable from the access to an MLaaS API. We want to emphasize that cryptographic proofs do not guarantee security in the case of side channel attacks, for example timing attacks.
Also other attacks on classifiers than the ones discussed here exist, for example evasion attacks. Such evasion attacks were found to be hard to defend~\cite{CW17}. Current state-of-the-art defenses are however applied before inference~\cite{MMSTV17,RSL18,HKWW17}, and can thus be integrated into \system without overhead. Additionally, any defense or mitigation in \system is transparent, as the code can be inspected, and thus \system does not rely on security by obscurity. 

Finally, in \system, the detection of an attack is promising: in contrast to the standard MLaaS setting, an enclave is tied to a particular person. It is hence possible to identify a user who submitted malicious data. 
Additionally, setting up a fresh enclave requires some effort. This implies that the service provider is actually able to persecute or expel clients who are found out to run attacks.

\section{Conclusion}
\label{sec:conclusion}

We have presented a novel deployment mechanism for ML models. It provides the same level of security of and control over the model as conventional server-side MLaaS execution. At the same time, it provides perfect privacy of the user data as it never leaves the client. In addition, we show the extensibility of our approach and how it facilities a range of features from pay-per-view monetization to advanced model protection mechanisms -- including the very latest work on model stealing and reverse engineering. 

We believe that this is an important step towards the overall vision of data privacy in machine learning \cite{PMSW18} as well as secure ML and AI \cite{SSPPMKJJHGGGCA17}. Beyond the presented work and direct implications on data privacy and model security -- this line of research implements another line of defense that in the future can help to tackle several problems in security related issues of ML that the community has been struggling to make sustainable progress. For instance, a range of attacks from membership inference, reverse engineering to adversarial perturbations rely on repeated queries to a model. Our deployment mechanism provides a framework that is compatible with wide spread use of ML models -- but yet can control or mediate access to the model directly (securing the model) or indirectly (advanced protection against inference attacks).

\balance
\bibliographystyle{ACM-Reference-Format}
\bibliography{main.bib}
\end{document}